\documentclass{article}
\usepackage{amsmath,graphicx,spconf}
\usepackage{amssymb}
\usepackage{acronym}
\usepackage{cite}
\usepackage{caption}
\usepackage{standalone}
\usepackage{subcaption}

\usepackage{verbatim}
\usepackage{soul, xcolor}

\usepackage[ruled,linesnumbered]{algorithm2e}  
\SetKwInput{KwData}{\textbf{Init}} 

\usepackage{KalmanNet}

\newtheorem{theorem}{Theorem}

\newcommand{\myH}{\tilde{\gvec{H}}}

\copyrightnotice{\copyright}

\toappear{}

\title{Uncertainty in Data-Driven Kalman Filtering for Partially known State-Space Models}

\name{Itzik Klein, Guy Revach, Nir Shlezinger, Jonas E. Mehr, Ruud J. G. van Sloun, and Yonina. C. Eldar
\thanks{
I. Klein is with the Hatter Department of Marine Technologies, University of Haifa, Israel (email: kitzik@univ.haifa.ac.il). 
G. Revach and J. E. Mehr are with the Institute for Signal and Information Processing (ISI), D-ITET, ETH Zürich, (email: grevach@ethz.ch, mehrjo@student.ethz.ch).
N. Shlezinger is with the School of ECE, Ben-Gurion University of the Negev, Beer Sheva, Israel (e-mail: nirshl@bgu.ac.il).
R. J. G. van Sloun is with the EE Dpt., Eindhoven University of Technology, and with Phillips Research, Eindhoven,  The Netherlands (e-mail: r.j.g.v.sloun@tue.nl).
Y. C. Eldar is with the Faculty of Math and CS, Weizmann Institute of Science, Rehovot, Israel (e-mail: yonina.eldar@weizmann.ac.il).
The authors thank Prof. Hans-Andrea Loeliger for his helpful comments and discussion.
}}

\address{}

\begin{document}

\maketitle
%
%
\begin{abstract} 
Providing a metric of uncertainty alongside a state estimate is often crucial when tracking a dynamical system. Classic state estimators, such as the \ac{kf}, provide a time-dependent uncertainty measure from knowledge of the underlying statistics; however,  \acl{dl} based tracking systems struggle to reliably characterize uncertainty. 
In this paper, we investigate the ability of \acl{kn}, a recently proposed; hybrid; \acl{mb}; deep state tracking algorithm, to estimate an uncertainty measure.
By exploiting the interpretable nature of \acl{kn}, we show that the error covariance matrix can be computed based on its internal features, as an uncertainty measure. We  demonstrate that  when the system dynamics are known, \acl{kn}—which learns its mapping from data without access to the statistics—provides uncertainty similar to that provided by the \ac{kf}; and while in the presence of evolution model-mismatch, \acl{kn} provides a more accurate error estimation.

%
\end{abstract}
\acresetall
%
%
\section{Introduction}\label{sec:intro}
Tracking a hidden state vector from noisy observations in \acl{rt} is at the core of many \acl{sp} applications. A leading approach to the task is the \ac{kf} \cite{kalman1960new}, which operates with low complexity and achieves a \ac{mmse} in setups characterized by \acl{lg} \ac{ss} models. A key merit of the \ac{kf}, which is of paramount importance in safety critical applications, e.g., autonomous driving, aviation, and medical, is its ability to provide uncertainty alongside state estimation \cite[Ch. 4]{durbin2012time}. \ac{kf} and its variants are \ac{mb} algorithms, and are therefore sensitive to inaccuracies in modeling the system dynamics using the \ac{ss} model. 

The recent success of deep learning architectures, such as \acp{rnn} \cite{chung2014empirical} and attention mechanisms \cite{vaswani2017attention}, in learning from complex time-series data in unstructured environments and in a model-agnostic manner, evoked interest in using them for tracking dynamical systems. However, the lack of interpretability of deep architectures and their inability to capture uncertainty \cite{becker2019recurrent}, together with the fact that they require many trainable parameters and large data sets even for simple  setups \cite{zaheer2017latent}, limit their applicability for safety-critical applications in hardware-limited systems. 

Characterizing uncertainty in \acp{dnn} is an active area of research \cite{nguyen2015deep, poggi2017quantitative, osband2021epistemic}. One approach is to use Bayesian \acp{dnn} \cite{jospin2020hands}, which, when combined with \ac{ss} models, enables extraction of the uncertainty. See, e.g., \cite{karl2016deep, krishnan2017structured, naesseth2018variational}. However, doing so relies on variational inference, which makes learning more complex and less scalable, and cannot be used directly for state estimation \cite{becker2019recurrent}. Alternatively, one can use deep learning techniques to estimate the \ac{ss} model parameters and then plug them into a variant of a \ac{kf} that provides uncertainty; e.g., \cite{abbeel2005discriminative, haarnoja2016backprop, laufer2018hybrid, xu2021ekfnet}. These approaches are limited in accuracy and complexity, requiring linearization of the second-order moments to compute the error covariance as required by the \ac{mb} \ac{kf}. The \ac{kf}-inspired \ac{rnn} proposed \cite{becker2019recurrent} was trained to predict both the state and the error. However, \cite{becker2019recurrent} focused on specific factorizable \ac{ss} models with partially observable states, for which the \ac{kf} is simplified to parallel scalar operations, and the proposed architecture does not naturally extend to general \ac{ss} models.

The recently proposed \acl{kn}  \cite{KalmanNetTSPa} utilizes \acp{rnn} to enhance the robustness of the \ac{kf} for complex and mismatched models as a form of \ac{mb} deep learning \cite{shlezinger2020model}. \acl{kn} was shown to operate reliably in practical hardware-limited systems \cite{KalmannetICAS21} due to its principled incorporation of partial domain knowledge. 
In this work we show how \acl{kn} can be extended to provide uncertainty measures,  
by exploiting its interpretable architectures and the fact that it preserves the internal flow of the \ac{kf}. In particular, we build upon the identification of an internal feature as the estimated \ac{kg}, and show that when combined with partial domain knowledge, it can be used to compute the time-dependent error {covariance} matrix. We numerically show that the extracted uncertainty indeed reflects the performance of \acl{kn}, providing similar state estimates and error measures as the \ac{kf}—which knows the \ac{ss} model—while being notably more reliable in terms of tracking and uncertainty in the presence of mismatched models.

The rest of the paper is organized as follows: Section~\ref{sec:Model} reviews the \ac{ss} model and recalls the \ac{kf} and \acl{kn}. Section~\ref{sec:KNet} analyzes the extraction of uncertainty in \acl{kn}, while Section~\ref{sec:NumEval} presents a numerical study. 
%
%
\section{System Model and Preliminaries}\label{sec:Model}
We review the \ac{ss} model and briefly recall the \ac{mb} \ac{kf}, since its operation serves as the baseline for \acl{kn} and for its uncertainty extraction scheme detailed in Section~\ref{sec:KNet}. We then recap the \acl{dd} filtering problem and the architecture of \acl{kn}. For simplicity, we focus on linear \ac{ss} models, although the derivations can also be used for non-linear models in the same manner as the extended \ac{kf} \cite[Ch. 10]{durbin2012time}. 

\subsection{System Model and Model-Based Kalman Filtering}\label{ssec:ssmdl} 
We consider a dynamical system characterized  by a linear, Gaussian, continuous evolution model in \acl{dt}. For $t\in\gint$, this \ac{ss} model is defined by \cite{bar2004estimation}
\begin{subequations}\label{eq:ssmodel}
\begin{align}\label{eqn:stateEvolution}
\gvec{x}_{t}&= 
\gvec{F}\cdot{\gvec{x}_{t-1}}+\gvec{w}_{t},& 
\gvec{w}_t\sim
\mathcal{N}\brackets{\gvec{0},\gvec{Q}},&
\quad
\gvec{x}_{t}\in\greal^m,\\ \label{eqn:stateObservation}
\gvec{y}_{t}&=
\gvec{H}\cdot{\gvec{x}_{t}}+\gvec{v}_{t},& 
\gvec{v}_t\sim
\mathcal{N}\brackets{\gvec{0},\gvec{R}},&
\quad
\gvec{y}_{t}\in\greal^n.    
\end{align}
\end{subequations}
In \eqref{eqn:stateEvolution}, $\gvec{x}_{t}$ is the latent state vector of the system at time $t$, which evolves by a linear state evolution matrix $\gvec{F}$ and by an \ac{awgn} $\gvec{w}_t$ with noise covariance $\gvec{Q}$. In \eqref{eqn:stateObservation}, $\gvec{y}_{t}$ is the vector of observations at time $t$, $\gvec{H}$ is the measurement matrix, and $\gvec{v}_t$ is an \ac{awgn} with measurement noise covariance $\gvec{R}$. The filtering problem deals with \acl{rt} \acl{se}; i.e., the recovery of $\gvec{x}_t$ from $\{\gvec{y}_\tau\}_{\tau \leq t}$ for each time instance $t$ \cite{durbin2012time}. 

The \ac{kf} is a two-step, low complexity, recursive algorithm that produces a new estimate $\hat{\gvec{x}}_t$ from a new observation $\gvec{y}_{t}$ based on the previous estimate $\hat{\gvec{x}}_{t-1}$ as a sufficient statistic. In the first step it {predicts} the statistical moments based on the previous \textit{a posteriori} estimates:
\begin{subequations}
\label{eqn:predict2}
\begin{align} 
\hat{\gvec{x}}_{t\given{t-1}} &= 
\gvec{F}\cdot{\hat{\gvec{x}}_{t-1}},
\quad
\mySigma_{t\given{t-1}} \!=\!
\gvec{F}\cdot\mySigma_{t-1}\cdot\gvec{F}^\top\!+\!\gvec{Q},\\\label{eq:s}
\hat{\gvec{y}}_{t\given{t-1}} &=
\gvec{H}\cdot{\hat{\gvec{x}}_{t\given{t-1}}},
\hspace{0.075cm}
\gvec{S}_{t\given{t-1}}\! =\!
\gvec{H}\cdot\mySigma_{t\given{t-1}}\cdot\gvec{H}^\top\!+\!\gvec{R}.
\end{align}
\end{subequations} 
In the second step, the \textit{a posteriori} moments are updated based on the \textit{a priori} moments. The pivot computation for this step is the \acl{kg} $\Kgain_t$:
\begin{equation}\label{eq:lingain}
\Kgain_{t}={\mySigma}_{t\given{t-1}}\cdot{\gvec{H}}^\top\cdot{\gvec{S}}^{-1}_{t\given{t-1}}.
\end{equation}
Given the new observation $\gvec{y}_t$, the state estimate, i.e., the first-order posterior, is obtained via
\begin{equation}\label{eq:update1}
\hat{\gvec{x}}_{t}=
\hat{\gvec{x}}_{t\given{t-1}}+\Kgain_{t}\cdot\Delta\gvec{y}_t,
\quad
\Delta\gvec{y}_t=\gvec{y}_t-\hat{\gvec{y}}_{t\given{t-1}}.
\end{equation}
The second-order posterior, which is the estimation error covariance, is then computed as
\begin{equation}\label{eq:update2}
{\mySigma}_{t}=
{\mySigma}_{t\given{t-1}}-\Kgain_{t}\cdot{\mathbf{S}}_{t\given{t-1}}\cdot\Kgain^{\top}_{t}.
\end{equation}
%
%
When the noise is Gaussian and the \ac{ss} model parameters  \eqref{eq:ssmodel} are fully known, the \ac{kf} is the \ac{mmse} estimator. 
%
%
%
\subsection{Data-Driven Filtering with KalmanNet}\label{subsec:KalmanNet}
In practice, the state evolution model \eqref{eqn:stateEvolution} is determined by the complex dynamics of the underlying system, while the observation model \eqref{eqn:stateObservation} is dictated by the type and quality of the observations. For instance, $\gvec{x}_t$ can be the location, velocity, and acceleration of a vehicle, while $\gvec{y}_{t}$ are measurements obtained from   several sensors. For \acl{rw} problems it is often difficult to characterize the \ac{ss} model accurately. In data-driven filtering, one relies on a {labeled} \acl{ds}—i.e., $\gvec{y}_t$ with its corresponding \acl{gt} $\gvec{x}_t$—to fill the information gap.

\acl{kn} \cite{KalmanNetTSPa} is a \ac{dnn}-aided architecture for \acl{se} and filtering with partial domain knowledge. It considers filtering without knowledge of the noise covariance matrices and with a possibly inaccurate description of the evolution matrix $\gvec{F}$ and  the observation model $\gvec{H}$, obtained from, e.g., understanding of the system dynamics and the sensing setup.  
Although $\gvec{R}$, the observation covariance, is not needed for state estimation with \acl{kn}, for simplicity of derivation, here we assume that it is known (or estimated) and thus set $\gvec{R} = \gvec{I}_n$, since one can always process $\gvec{R}^{-1/2}\gvec{y}_t$ without altering the achievable  \ac{mse}. 

\acl{kn} \cite{KalmanNetTSPa} implements data-driven filtering by augmenting the theoretically solid flow of the \ac{mb} \ac{kf} with a \ac{rnn}. The latter is designed to learn to estimate the \ac{kg}, whose \ac{mb} computation encapsulates the missing domain knowledge, from data. The inherent memory of the \ac{rnn} allows implicit tracking of the second-order statistical moments without requiring knowledge of the underlying noise statistics. In a similar manner as the \ac{kf}, it predicts the first-order moments  $\hat{\gvec{x}}_{t\given{t-1}} $ and 
$\hat{\gvec{y}}_{t\given{t-1}}$, which are then used to estimate $\gvec{x}_t$ via  \eqref{eq:update1}. The overall system, illustrated in Fig.~\ref{fig:KNet1}, is trained end-to-end to minimize the \ac{mse}  between $\hat{\gvec{x}}_t$ and the true $\gvec{x}_t$. 
%
%
\begin{figure}[]
\includegraphics[width=1\columnwidth]{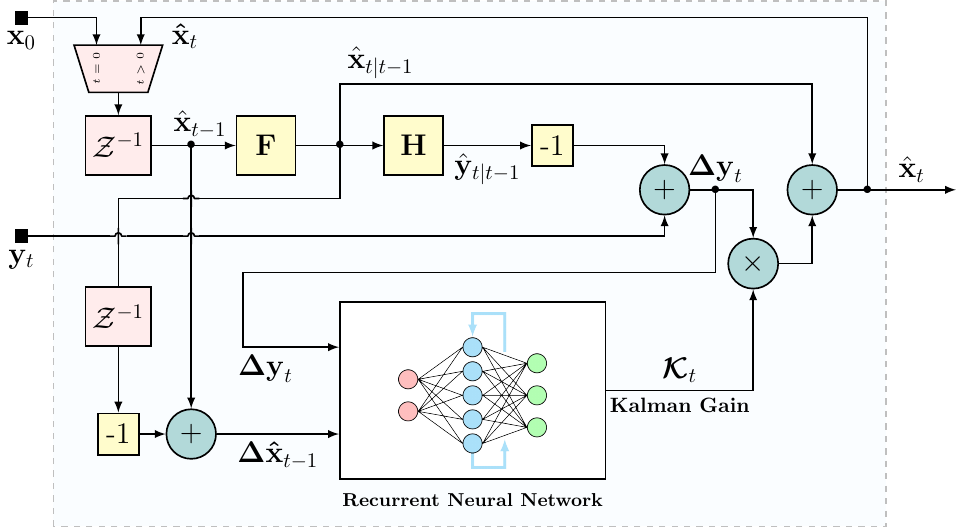}
\caption{\acl{kn} block diagram.}
\label{fig:KNet1}
\end{figure}
%
%
\section{Uncertainty in KalmanNet}\label{sec:KNet}
\acl{kn}, detailed in Subsection \ref{subsec:KalmanNet}, is designed and trained to estimate the state variable $\gvec{x}_t$. Neither its architecture nor the loss measure it uses for training encourage \acl{kn} to maintain an estimate of its error covariance, which the \ac{kf} provides. However, as we show here, the fact that \acl{kn} preserves the flow of the \ac{kf} allows to estimate its error covariance from its internal features.

%
%
\subsection{Kalman Gain-based Error Covariance}
To extend \acl{kn} to provide uncertainty, we build on top of the interpretable feature of \acl{kn}, which is the estimated \ac{kg}. Combined with the observation model, the \ac{kg} can be used to compute the time-dependent, error covariance matrix $\mySigma_t$, thus bypassing the need to explicitly estimate the evolution model, as stated in the following theorem: 
\begin{theorem}
\label{thm:Covariance}
Consider the \ac{ss} model \eqref{eq:ssmodel} where $\gvec{H}$ has full column rank; i.e.,  $\myH=\brackets{\gvec{H}^\top\cdot\gvec{H}}^{-1}$ exists. Then, filtering with the \ac{kg} $\Kgain_t$ \eqref{eq:lingain} results in estimation with error covariance  
\begin{align}
{\mySigma}_{t}=\brackets{\gvec{I}_m-\Kgain_t\cdot\gvec{H}}&\cdot \myH\cdot
\gvec{H}^\top\brackets{\gvec{I}_n-\gvec{H}\cdot\Kgain_{t}}^{-1}\notag \\
&\cdot\gvec{H}\cdot\Kgain_{t}\cdot\gvec{H}\cdot\myH. 
\label{eqn:KN_covariance}
\end{align}
\end{theorem}
\begin{proof}
By combining \eqref{eq:lingain} and \eqref{eq:update2}, the error covariance can be written as ${\mySigma}_{t}=\brackets{\gvec{I}_m-\Kgain_t\cdot\gvec{H}}{\mySigma}_{t|t-1}$. Thus, to estimate ${\mySigma}_{t}$, we express ${\mySigma}_{t|t-1}$ using $\Kgain_t$ and the available domain knowledge. %
%
%
To that end, multiplying   \eqref{eq:lingain} by $\gvec{H}$ gives
\begin{equation}\label{eq:step2}
\gvec{H}\cdot\Kgain_{t}=\gvec{H}\cdot{\mySigma}_{t\given{t-1}}\cdot{\gvec{H}}^\top\cdot{\gvec{S}}^{-1}_{t\given{t-1}}.
\end{equation}
Next, \eqref{eq:step2} is multiplied by $\gvec{S}_{t\given{t-1}}$ from the right side, followed by substitution of \eqref{eq:s}, which yields
\begin{equation}\label{eq:step3}
\gvec{H}\cdot\Kgain_{t}\cdot\brackets{\gvec{H}\cdot\mySigma_{t\given{t-1}}\cdot\gvec{H}^\top+\gvec{R}}=
\gvec{H}\cdot{\mySigma}_{t\given{t-1}}\cdot{\gvec{H}}^\top.
\end{equation}
Combining terms in \eqref{eq:step3} results in
\begin{equation} \label{eq:step5}
\gvec{H}\cdot\mySigma_{t\given{t-1}}\cdot\gvec{H}^\top=
\brackets{\gvec{I}_n-\gvec{H}\cdot\Kgain_{t}}^{-1}\cdot\gvec{H}\cdot\Kgain_{t}.
\end{equation}
Multiplying \eqref{eq:step5} from the left  by $\gvec{H}^\top$ and from the right by $\gvec{H}$ yields 
$\gvec{H}^\top\cdot\gvec{H}\cdot\mySigma_{t\given{t-1}}\cdot\gvec{H}^\top\cdot\gvec{H}= \gvec{H}^\top\brackets{\gvec{I}_n-\gvec{H}\cdot\Kgain_{t}}^{-1}\cdot\gvec{H}\cdot\Kgain_{t}\cdot\gvec{H}$, 
which, when $\myH=\brackets{\gvec{H}^\top\cdot\gvec{H}}^{-1}$ exists, results in  
\begin{equation}
    \mySigma_{t\given{t-1}}=\myH\cdot
\gvec{H}^\top\brackets{\gvec{I}_n-\gvec{H}\cdot\Kgain_{t}}^{-1}\cdot \gvec{H}\cdot\Kgain_{t} \cdot\gvec{H}\cdot\myH,
\end{equation}
concluding the proof of the theorem.
\end{proof}
Theorem~\ref{thm:Covariance} indicates that when the observation model is known and $\gvec{H}$ has full column rank—requiring that the  number of measurements not be smaller than the number of tracked states—then one can extend \acl{kn} to predict its error covariance alongside its state estimation. The resulting procedure is summarized as Algorithm~\ref{alg:Algo1}. 

%
\subsection{Discussion}\label{ssec:discussion}
%
%
\begin{algorithm}[t!]
\caption{\acl{kn} with error prediction}
\label{alg:Algo1}
\KwData{Trained \acl{kn}, previous estimate $\hat{\gvec{x}}_{t-1}$, incoming observations $\gvec{y}_t$;  }
{\bf Filtering:} Apply \acl{kn} to estimate $\hat{\gvec{x}}_{t}$\;
{\bf Error prediction:} Estimate $\hat{\mySigma}_t$ via \eqref{eqn:KN_covariance} with $\Kgain_t$ being the output of the internal \ac{rnn} of \acl{kn}\;
\KwOut{Predicted state $\hat{\gvec{x}}_{t}$ and error $\hat{\mySigma}_t$}
\end{algorithm}
\acl{kn} \cite{KalmanNetTSPa} was designed for \acl{rt} tracking in \acl{nl} \ac{ss} models with unknown noise statistics. Algorithm~\ref{alg:Algo1} extends it to provide not only $\hat{\gvec{x}}_t$, an estimate of the state—i.e., for which it was previously derived and trained—but also $\hat{\mySigma}_t$, an estimate of the time-dependent covariance matrix as an uncertainty measure. This is achieved since \acl{kn} retains the interpretable flow of the \ac{mb} \ac{kf}, where the \ac{kg} is learned by a dedicated \ac{rnn}. Numerical evaluations presented in Section~\ref{sec:NumEval} indicate that \acl{kn} computes the true covariance matrix as it reflects the true empirical error—as expected from the \ac{mse} bias variance decomposition theorem \cite{friedman2001elements}.

For simplicity and clarity of exposition, Algorithm~\ref{alg:Algo1} was derived for linear \ac{ss} models. In Section~\ref{sec:NumEval} we also present results for a \acl{nl} chaotic system, where the {extended} \ac{kf} derivation is used;  i.e., the \ac{ss} model matrices are replaced with their respecting Jacobians \cite[Ch. 10]{durbin2012time}. The additional requirement—that observation noise covariance $\gvec{R}$ is known and needed only for the purpose of estimating the covariance—is often satisfied, as in many applications the main modelling challenge is related to the state evolution rather than the measurement noise. Since we assume that we have access to a {labeled} \acl{ds}, one can estimate $\gvec{R}$ with standard techniques without assumptions regarding the evolution. The case  where $\gvec{H}$ is not full column rank is left for future work.
%
%
\section{Numerical Evaluations}\label{sec:NumEval}
In this section we numerically\footnote{{The source code along with additional information on the numerical study can be found online at \url{https://github.com/KalmanNet/ERRCOV_ICASSP22}.}} compare the covarince computed by \acl{kn} to the one computed by the \ac{kf} in the linear and \acl{nl} cases.

We start with the \acl{lg} \ac{ss} model, for which the \ac{kf} achieves the \ac{mmse} lower bound, and is an {unbiased} estimator \cite{humpherys2012fresh}. We generate data from a scalar \ac{ss} model where $\gvec{H}=\gvec{Q}=\gvec{R}=1$ and $\gvec{F}=0.9$, with sequence length of $100$ samples, and compare the error predicted by \acl{kn} and its deviation from the true error to those computed by the \ac{mb} \ac{kf} that knows the \ac{ss} model. We observe in Fig.~\ref{fig:Linear_full_err_evo} that the theoretical error computed by \ac{kf} for each time step $t$ and the error predicted by \acl{kn} coincide, and that both algorithms have a similar empirical error. In Fig.~\ref{fig:Linear_full_traj} we demonstrate that for a single realization of \acl{gt} trajectory, both algorithms produce the same uncertainty bounds. Next, we consider the case where both \acl{kn} and the \ac{mb} \ac{kf} are plugged in with a mismatched model parameter $\gvec{F}=0.5$. In Figs.~\ref{fig:Linear_par_err_evo} and \ref{fig:Linear_par_traj} we observe that for such a model mismatch, \acl{kn} produces uncertainty similar to the empirical error, while the \ac{kf} underestimates its empirical error.
%
%
\begin{figure}[t!]
\centering
\includegraphics[width=1\columnwidth]{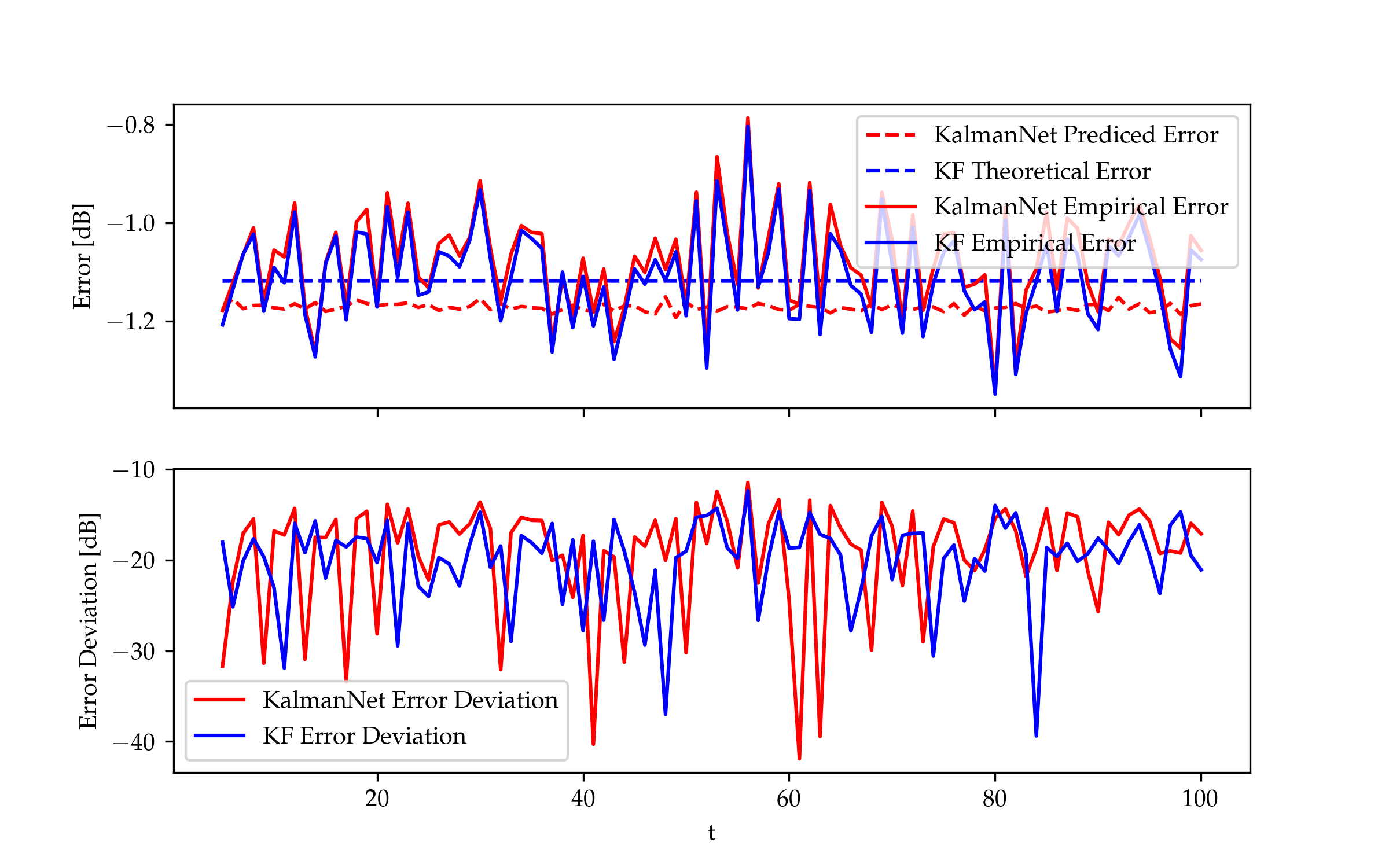}
\caption{Average error, linear \ac{ss} model. }
\label{fig:Linear_full_err_evo}
\vspace{-0.2cm}
\end{figure}
%
%
\begin{figure}[t!]
\centering
\includegraphics[width=1\columnwidth]{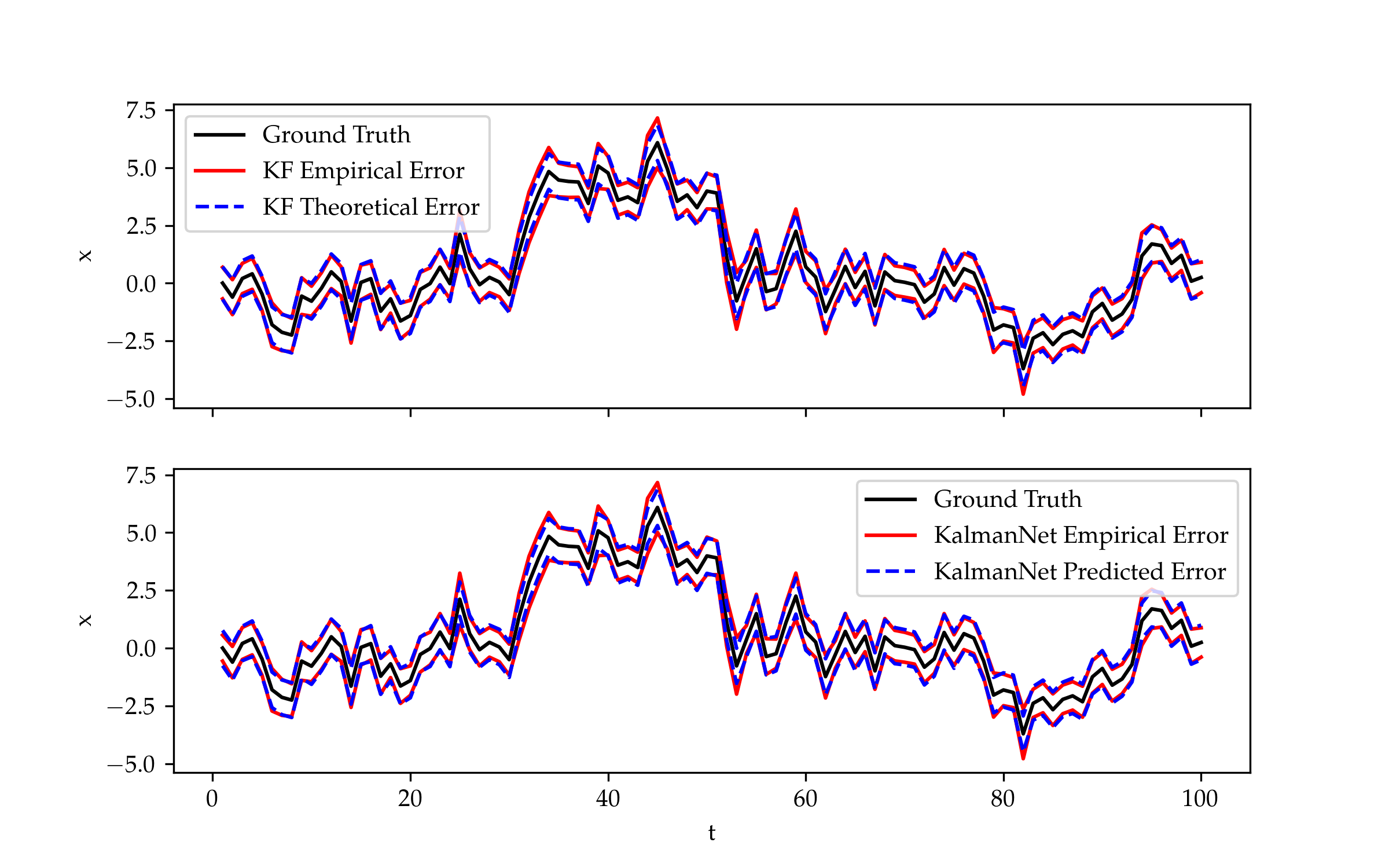}
\caption{Single trajectory, linear \ac{ss} model.}
\label{fig:Linear_full_traj}
\vspace{-0.2cm}
\end{figure}
%
%
\begin{figure}[t!]
\centering
\includegraphics[width=1\columnwidth]{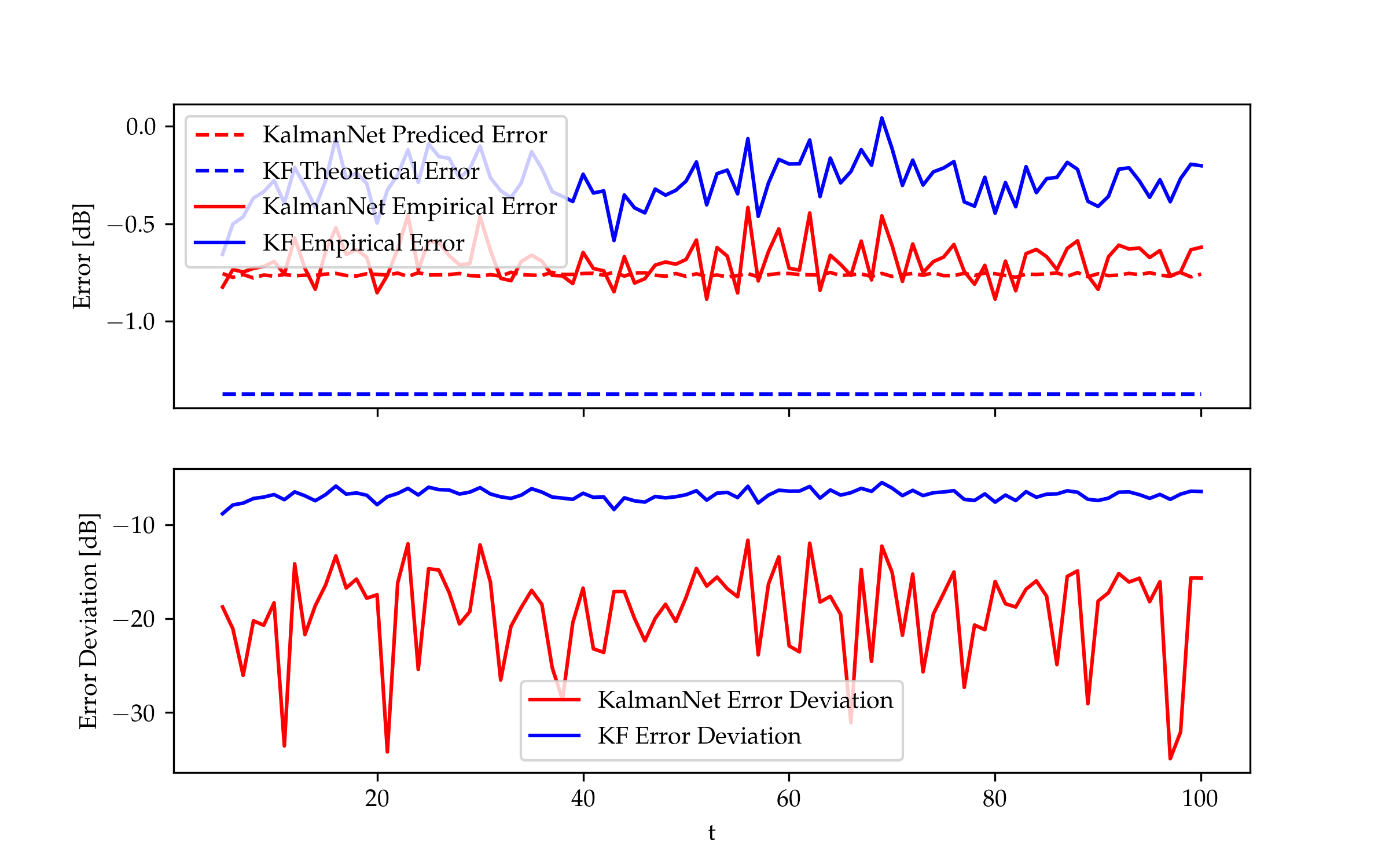}
\caption{Average error, linear \ac{ss} model with mismatch.}
\label{fig:Linear_par_err_evo}
\vspace{-0.2cm}
\end{figure}
%
%
\begin{figure}[t!]
\centering
\includegraphics[width=1\columnwidth]{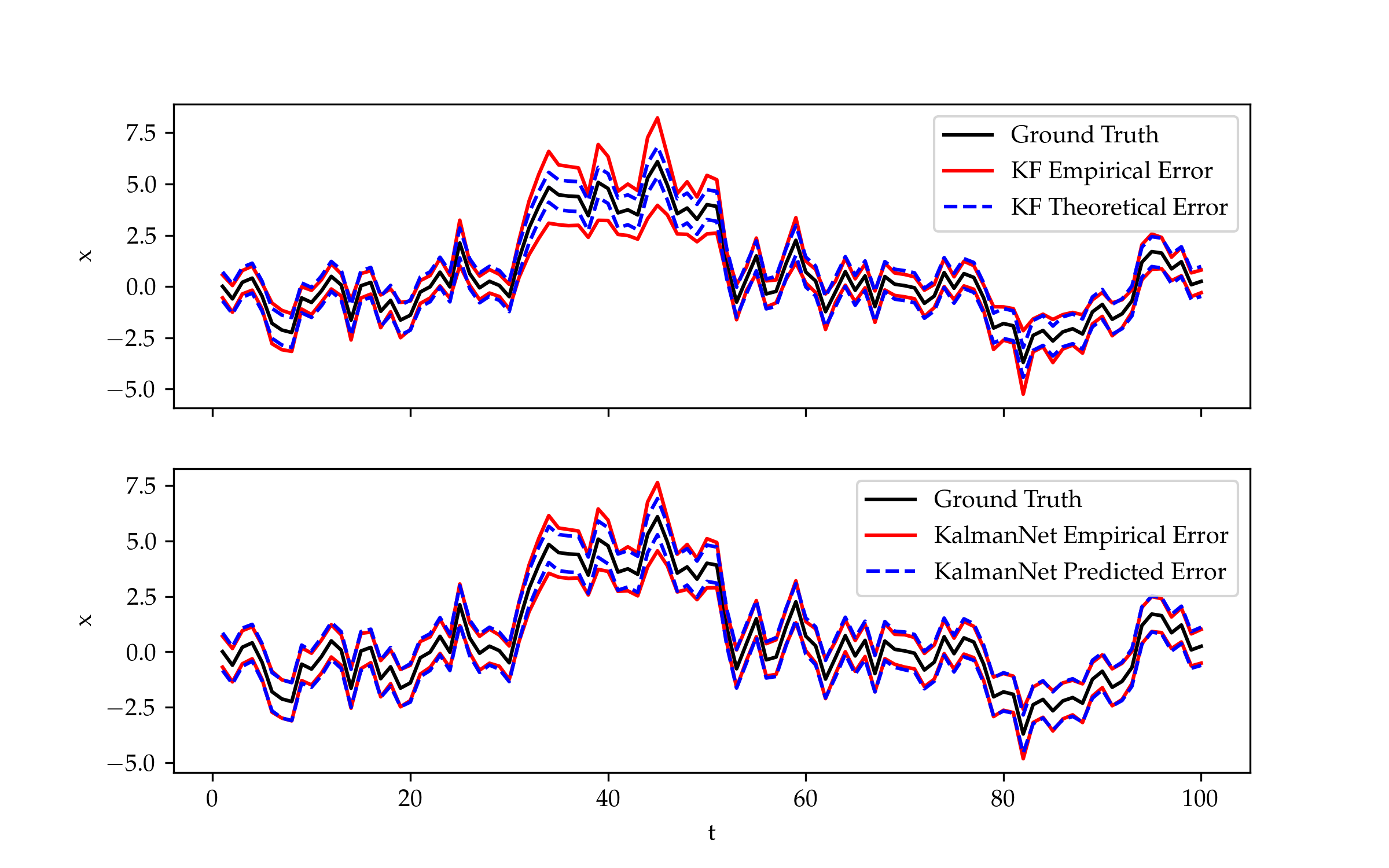}
\caption{Single trajectory, linear \ac{ss} model with mismatch.}
\label{fig:Linear_par_traj}
\vspace{-0.2cm}
\end{figure}

Next, we demonstrate the merits of \acl{kn} when filtering the \acl{la}—a challenging three-dimensional \acl{nl} chaotic system—and compare its performance to the extended \ac{kf}. See \cite{KalmanNetTSPa} for a detailed description of this setup. The model mismatch in this case is due to sampling a \acl{ct} system characterized by differential equations to \acl{dt}. In Fig.~\ref{fig:Lorentz} we  clearly observe that \acl{kn} achieves a lower \ac{mse} and estimates its error fairly accurately, while the extended \ac{kf} overestimates it. 
%
%
\begin{figure}[t!]
\centering
\includegraphics[width=1\columnwidth]{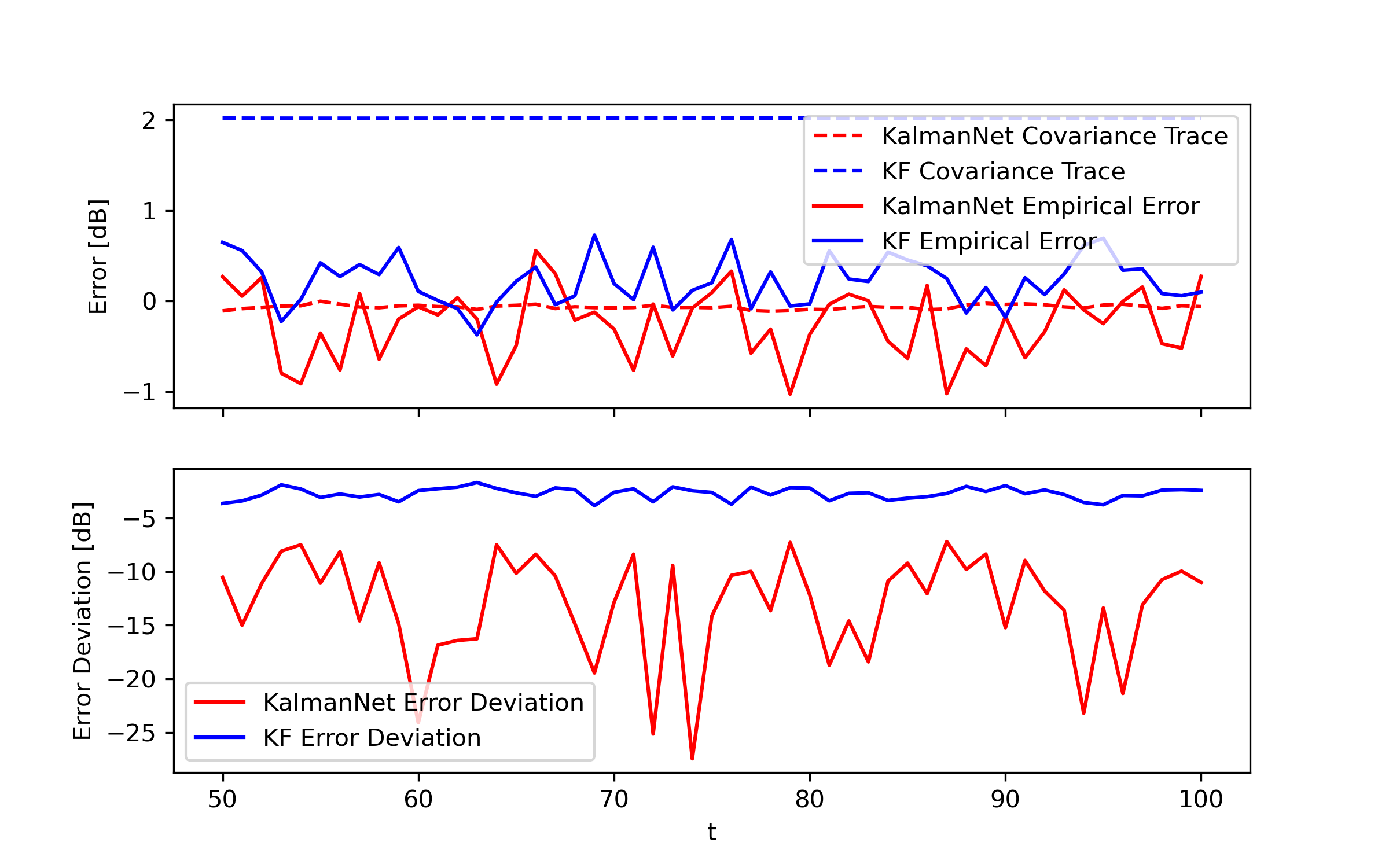}
\caption{Average error, \acl{nl} \acl{la} model.}
\label{fig:Lorentz}
\vspace{-0.2cm}
\end{figure}
\section{Conclusions}\label{sec:Conclusions}
In this work we extended the recently proposed \acl{kn} state estimator to predict its error alongside the latent state. This is achieved by exploiting the hybrid \acl{mb}/\acl{dd} architecture of \acl{kn}, which produces the \ac{kg} as an internal feature. We prove that one can often utilize the learned \ac{kg} to predict the error covariance as a measure of uncertainty. Our numerical results demonstrate that this extension allows \acl{kn} to accurately predict both the state and error, improving upon the \ac{kf} in the presence of model-mismatch and non-linearities. 
%
\bibliographystyle{IEEEbib}
\bibliography{IEEEabrv,KalmanNet}

\end{document}